\documentclass[journal,draftcls, onecolumn]{IEEEtran}
\usepackage[dvipdfmx]{graphicx}
\usepackage{amsmath}
\usepackage{amssymb}
\usepackage{amsfonts}
\usepackage{mathrsfs}
\usepackage{theorem}
\usepackage{color}
\DeclareMathAlphabet{\bm}{OML}{cmm}{b}{it}
\usepackage[ruled,vlined]{algorithm2e}
\usepackage{algorithmic}

\theorembodyfont{\rmfamily}
\newtheorem{theorem}{Theorem}
\newtheorem{lemma}[theorem]{Lemma}
\newtheorem{definition}[theorem]{Definition}
\newtheorem{corollary}[theorem]{Corollary}
\newtheorem{remark}[theorem]{Remark}
\newtheorem{proposition}[theorem]{Proposition}
\newtheorem{example}[theorem]{Example}

\newcommand{\bol}[1]{\mathbf{#1}}

\newcommand{\textchange}[1]{#1}

\begin{document}

\title{
Interval Algorithm for Random Number Generation: Information Spectrum Approach\thanks{A part of this paper is submitted for possible presentation at
2019 IEEE Information Theory Workshop.}
}

\author{Shun Watanabe and Te Sun Han}


\maketitle
\begin{abstract}
The problem of exactly generating a general random process (target process) by using another general
random process (coin process) is studied. 
The performance of the interval algorithm, introduced by Han and Hoshi, is
analyzed from the perspective of information spectrum approach. 
When either the coin process or the target process has one point spectrum,
the asymptotic optimality of the interval algorithm among any random number generation algorithms is proved,
which demonstrates utility of the interval algorithm beyond the ergodic process. 
Furthermore, the feasibility condition of exact random number generation is also elucidated.
Finally, the obtained general results are illustrated by 
the case of generating a Markov process from another Markov process.
\end{abstract}

\section{Introduction} \label{section:introduction}

We revisit the problem of exactly generating a random process, termed  {\em target process},
from another random process, termed {\em coin process}.
This problem has a long history. In a seminal paper \cite{Neumann51}, von Neumann introduced an algorithm to generate 
the independently identically distributed (i.i.d.) binary unbiased process from
an i.i.d. binary biased process. 
Subsequently, his result was extended and refined in various directions \cite{Samuelson:68, HoeSim:70, Elias:72, Blum86, Peres:92}. 
On the other hand, Knuth and Yao \cite{KnuYao:76} studied the problem of generating an arbitrary target process
using i.i.d. unbiased coin process. Later, the problem of generating an arbitrary target process from an arbitrary coin process 
was studied by various researchers \cite{Roche:91, Abrahams:96}. 
For instance, by generalizing the approach in \cite{KnuYao:76}, Abrahams proposed an algorithm to generate an arbitrary target process  
from an i.i.d. biased (not necessarily binary) coin process \cite{Abrahams:96}; however, this algorithm is only applicable to
the algebraic coin, i.e., the case where the probabilities of coin random variable is described by the root of a polynomial equation.  
In this paper, we focus on
the {\em interval algorithm} proposed in \cite{HanHoshi97}. The interval algorithm is constructive, and
it can be applied to any coin/target processes that may have memory and may not be stationary nor ergodic.
Thus, it is of interest to identify under what circumstances the interval algorithm has the optimal performance.
In fact, despite simplicity of the algorithm, performance analysis of the interval algorithm is not straightforward.  

When the coin process is i.i.d., Han and Hoshi have shown that the interval algorithm 
asymptotically attains the optimal performance among any random number generation algorithm \cite{HanHoshi97};
more precisely, they have shown that the average stopping time of the coin process, i.e., the average number of coin tosses, of the interval
algorithm converges to the fundamental limit, which is given by the ratio between the entropy rates of the coin and target processes. 
Using representation of real numbers, Oohama refined Han and Hoshi's performance analysis of the interval algorithm \cite{oohama:11, Oohama:16}.

For i.i.d. coin processes, the performance of the interval algorithm is fairly well understood.
However, in practice, it is also desirable to use a coin process that has a memory, such as the Markov process.
When the coin process is Markov, the performance analysis of the interval algorithm become intractable. 
In fact, even though performance analysis of the interval algorithm for the Markov coin process was conducted in \cite{HanHoshi97, Oohama:16},
the analyses there do not guarantee asymptotic optimality. One of the motivations of this paper is to elucidate the performance of the interval algorithm
when the coin process is Markov. 
 
On the other hand, Uyematsu and Kanaya studied the overflow probability of the stopping time of the interval algorithm \cite{UyeKan99, UyeKan99b}.
In \cite{UyeKan99b}, they derived an exponential convergence rate of the overflow probability of the stopping time for i.i.d. processes.
In \cite{UyeKan99}, using the sample path approach \cite{Shields-book},
they derived almost sure convergence results on the stopping time for general coin/target processes; 
however, since their characterization is in terms of the quantities defined for sample path \cite{MurKan:99}, 
it is not immediately clear how to evaluate those quantities other than ergodic processes.
Moreover, they only analyzed the interval algorithm and did not discuss the optimality of the interval algorithm
among other random number generation algorithms.
Even though the almost sure convergence
analysis is of theoretical importance, the authors believe that the average performance analysis is preferable in practice
since it provides more insights on the finite length performance along the way of deriving asymptotic results.
It should be also pointed out that the almost sure convergence of stopping time does not immediately provide
performance guarantee of the average stopping time (cf.~Remark \ref{remark:implication-almost-sure}).
 
As a related problem to the above, the problem of random number generation with
approximation error has been extensively studied in the past few decades \cite{han:93, vembu:95, NagMiy:96, han:00, hayashi:08, AltWag:12, nomura:13, KumHay:17}. 
In such a direction of research, the information spectrum approach introduced in \cite{han:93, han:book} is successfully used to derive 
fairly general results. 

In this paper, we apply the information spectrum approach to the problem of exactly generating a random process
by another random process. First, we derive a converse bound on the overflow probability of the stopping time for 
any random number generation algorithms. Second, we derive an achievability bound on the overflow probability of the stopping time
that can be attained by the interval algorithm. Using these bounds, we examine the asymptotic optimality of the interval algorithm 
for general coin/target processes. For the criterion of the overflow probability of the stopping time, 
when either the coin or the target process has one point spectrum,
the optimality of the interval algorithm among any random number generation algorithms is proved.
For the average stopping time criterion, when the coin process has one point spectrum with an additional mild condition,
the optimality of the interval algorithm among any random number generation algorithms is proved. 
These results demonstrate the utility of the interval algorithm for non-stationary and/or non-ergodic processes. 
As a side result, we also elucidate the condition that exact random number generation is possible. 
Finally, we illustrate the obtained general results by the case of Markov coin/target processes.

The rest of the paper is organized as follows. In Section \ref{sec:problem}, we describe the 
problem formulation and derive a converse bound for any random number generation algorithm. 
In Section \ref{sec:performance-interval}, we derive an achievability bound for the interval algorithm.
In Section \ref{sec:asymptotic}, we conduct the asymptotic analysis. 
In Section \ref{Sec:connection}, we mention the connection between the variable-length
random number generation and the fixed-length random number generation.
We close the paper with discussion in Section \ref{sec:discussion}. 

\subsection*{Notation}

Throughout the paper, random variables (eg.~$X$) and their realizations (eg.~$x$) are denoted by capital
and lower case letters, respectively. The ranges of random variables are denoted by the respective calligraphic letters (eg.~${\cal X}$).
The probability distribution of random variable $X$ is denoted by $P_X$.
Similarly, $X^n=(X_1,\ldots,X_n)$ and $x^n=(x_1,\ldots,x_n)$ denote, respectively, a random vector and its realization
in the $n$th Cartesian product ${\cal X}^n$ of ${\cal X}$. We use the standard notations for information measures \cite{cover}, such as the entropy $H(X)$,
the min-entropy $H_{\min}(X) = \min_x \log \frac{1}{P_X(x)}$,
and the binary entropy function $h(t) = t \log (1/t) + (1-t)\log(1/(1-t))$ for $0\le t \le 1$. For a random process $\bm{X}=\{ X^n \}_{n=1}^\infty$, the spectral sup-entropy
and the spectral inf-entropy are denoted by  
\begin{align} \label{eq:notation-sup-entropy}
\overline{H}(\bm{X}) := \inf\bigg\{ \lambda : \lim_{n\to\infty} \Pr\bigg( \frac{1}{n} \log \frac{1}{P_{X^n}(X^n)} \ge \lambda \bigg) = 0 \bigg\}
\end{align}
and 
\begin{align} \label{eq:notation-inf-entropy}
\underline{H}(\bm{X}) := \sup\bigg\{ \lambda : \lim_{n\to\infty} \Pr\bigg( \frac{1}{n} \log \frac{1}{P_{X^n}(X^n)} \le \lambda \bigg) = 0 \bigg\},
\end{align}
respectively \cite{han:book}.
The sup-entropy rate is denoted by 
\begin{align} \label{eq:notation-entropy}
H(\bm{X}) := \limsup_{n\to\infty} \frac{1}{n}H(X^n),
\end{align}
and it coincides with 
the entropy rate if the limit exists. The base of $\log$ and $\exp$ is $2$ and the natural logarithm is denoted by $\ln$.

\section{Problem Formulation and Basic Results} \label{sec:problem}

In this section, we describe the problem formulation of random number generation 
with variable length coin tossing. Let $\bm{X} = \{ X^m = (X_1,\ldots,X_m) \}_{m=1}^\infty$ be a random process taking values in a finite 
set ${\cal X} = \{1,\ldots, M\}$, 
and let $\bm{Y} = \{ Y^n = (Y_1,\ldots,Y_n) \}_{n=1}^\infty$ be a random process taking values in a finite set ${\cal Y}=\{1,\ldots,N\}$.
Unless otherwise stated, the distributions $P_{X^m}$ and $P_{Y^n}$ of the processes can be arbitrary as long as they are consistent over time; \textchange{i.e.,
\begin{align*}
\sum_{x_{m+1}} P_{X^{m+1}}(x^m, x_{m+1}) = P_{X^m}(x^m)
\end{align*}
for every $m$ and $x^m \in {\cal X}^m$, and similarly for $P_{Y^n}$.}
We shall consider the problem of random number generation to
simulate the sequence of random variables $Y^n$ using outputs from the sequence of random variable $X^m$;
the former is referred to as {\em target process} and the latter is referred to as {\em coin process}.
Specifically, an algorithm of random number generation with variable length coin tossing is described by
a full $M$-ary tree of possibly infinite depth (see Example \ref{example:algorithm-tree} below); 
that is, it is described by a deterministic function
\begin{align} \label{eq:functional-algorithm}
\phi: \bigcup_{i=0}^\infty {\cal X}^i \to \{\bot\} \cup {\cal Y}^n
\end{align}
such that $\phi(x^i) \in {\cal Y}^n$ if $x^i$ corresponds to a leaf
and $\phi(x^i)=\bot$ otherwise, 
where $\bot$ is the null sequence and ${\cal X}^0 = \{ \bot \}$.
Let ${\cal L}_\phi$ be the set of all leaves, i.e.,
\begin{align*}
{\cal L}_\phi := \bigg\{ s \in \bigcup_{i=0}^\infty {\cal X}^i : \phi(s) \in {\cal Y}^n \mbox{ and for all proper prefix } s^\prime \mbox{ of } s,~\phi(s^\prime)=\bot \bigg\}.
\end{align*}
For a leaf $s=(x_1,\ldots,x_i) \in {\cal L}_\phi$, its depth $i$ is denoted by $|s|$.

For a given infinite sequence $x_1,x_2,\ldots$, starting with $i=0$, we output a symbol in ${\cal Y}^n$ by the following algorithm:
\begin{enumerate}
\item \label{general-algo:step1} If $\phi(x^i)$ is in ${\cal Y}^n$, output $\phi(x^i)$ and terminate;

\item \textchange{Set $i = i+1$}, and go back to Step \ref{general-algo:step1}.
\end{enumerate}
For performance analysis, it is convenient to consider the output
of the algorithm for an input sequence of finite length. By an abuse of notation, we denote the output of the above algorithm for a sequence $x^m$ by $\phi(x^m)$, i.e.,
$\phi(x^m) = y^n$ if the algorithm terminates with output $\phi(x^i)=y^n$ for some $i \le m$, and $\phi(x^m)=\bot$ otherwise. 
The stopping time of the algorithm, i.e., the minimum integer $m \ge 0$ such that $\phi(X^m) \in {\cal Y}^n$, is denoted by $T$;
\textchange{note that the stopping time $T$ is the random variable that is induced by the algorithm and the coin process $\bm{X}$.}
For any fixed length $n$ of the target process, we require that the probability law of the output of the algorithm coincides with $P_{Y^n}$ exactly as $m\to \infty$, i.e.,
\begin{align} \label{eq:validity}
\sum_{s \in {\cal L}_\phi: \atop \phi(s)=y^n} P_{X^{|s|}}(s) = \lim_{m\to\infty} \Pr( \phi(X^m) = y^n ) = P_{Y^n}(y^n)
\end{align}
for every $y^n \in {\cal Y}^n$.

Note that the algorithm described as in \eqref{eq:functional-algorithm} outputs sequence $y^n \in {\cal Y}^n$ of length $n$ collectively. Practically,
it is also important to consider an algorithm that outputs symbol $y_j$ whenever it is ready; such an algorithm is termed a sequential algorithm.
We will consider a sequential version of the interval algorithm in the next section. It should be noted that, for a given sequential algorithm, we can describe  that algorithm
in the form of \eqref{eq:functional-algorithm} by pooling $y_1,\ldots,y_{n-1}$ until $y_n$ is ready to be output. Thus, the converse bound to be described later in this section
is also valid for sequential algorithms.

Let us illustrate the problem formulation by the following simple example.
\begin{example}[\cite{cover}] \label{example:algorithm-tree}
Let us consider generation of one symbol, i.e., $n=1$, of random variable with distribution $P_Y=(2/3,1/3)$ using
the i.i.d. sequence $\{X^m\}_{m=1}^\infty$ of unbiased binary random variables. For this example, by noting the binary expansions
\begin{align*}
\frac{2}{3} &= \sum_{i=1}^\infty 2^{-(2i-1)}, \\
\frac{1}{3} &= \sum_{i=1}^\infty 2^{-2i},
\end{align*}  
we can construct an algorithm with the tree described in Fig.~\ref{Fig:example-algorithm-tree}.
\end{example}

\begin{figure}[t]
\centering{
\includegraphics[width=0.4\textwidth]{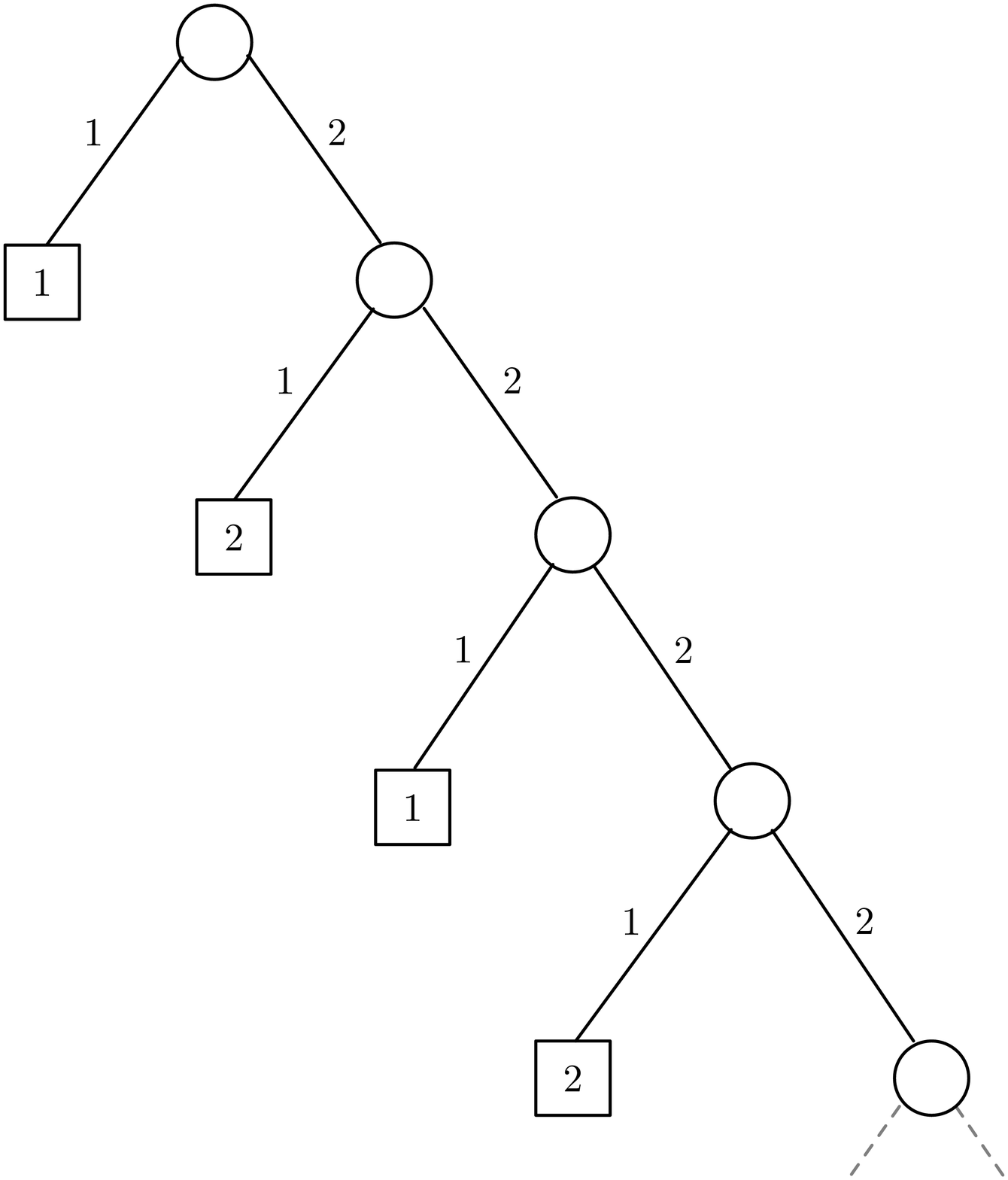}
\caption{A description of algorithm tree in Example \ref{example:algorithm-tree}. The label $1$ or $2$ of each edge indicates 
the outcome of coin random variable $X_i$. The circle node at depth $i$ indicates that the algorithm does not terminate 
when $X^i$ is observed; the square node at depth $i$ indicates that the algorithm terminates with the output labeled in that square node.}
\label{Fig:example-algorithm-tree}
}
\end{figure}

When the coin process is i.i.d. with common distribution $P_X$, there is a useful lower bound on the expected stopping time (cf.~\cite[Eq.~(2.4)]{HanHoshi97}):
\begin{align*}
\mathbb{E}[T] \ge \frac{H(Y^n)}{H(X)}.
\end{align*}
Since this lower bound is not available for general coin processes, the following lower bound
on the overflow probability of the stopping time $\Pr( T>m )$ 
is of importance in the latter sections; note that this lower bound is reminiscent of \cite[Lemma 2.1.2]{han:book}.

\begin{theorem} \label{theorem:converse}
For arbitrary random number generation algorithm \textchange{satisfying \eqref{eq:validity}} and integer $m \ge 0$, the overflow probability of the stopping time satisfies 
\begin{align}
\Pr( T > m ) &\ge P_{Y^n}({\cal T}_n^c(\tau)) - P_{X^m}({\cal S}_m(\lambda)) - 2^{-\tau+\lambda} \label{eq:single-shot-converse-1} \\
&= P_{X^m}({\cal S}_m^c(\lambda)) - P_{Y^n}({\cal T}_n(\tau)) - 2^{-\tau + \lambda} \label{eq:single-shot-converse-2}
\end{align}
for arbitrary real numbers $\tau,\lambda \ge 0$, where 
\begin{align}
{\cal S}_m(\lambda) &:= \bigg\{ x^m \in {\cal X}^m :  \log \frac{1}{P_{X^m}(x^m)} \ge \lambda \bigg\}, \label{eq:typical-X} \\
{\cal T}_n(\tau) &:= \bigg\{ y^n \in {\cal Y}^n : \log \frac{1}{P_{Y^n}(y^n)} \le \tau \bigg\}. \label{eq:typical-Y}
\end{align}
\end{theorem}
\begin{proof}
Without loss of generality, we can assume that there is no leaf $s \in {\cal L}_\phi$ such that $|s| < m$;
otherwise, we can expand that leaf to depth $m$ without changing the overflow probability $\Pr(T > m)$.
Thus, we assume this assumption is satisfied in the rest of the proof.

Let 
\begin{align*}
{\cal B} := \big\{ s \in {\cal L}_\phi : |s|=m \big\}.
\end{align*}
Then, we can write
\begin{align*}
\Pr( T > m) = \sum_{ s \in {\cal B}^c} P_{X^{|s|}}(s),
\end{align*}
where ${\cal B}^c = {\cal L}_\phi \backslash {\cal B}$. Let 
\begin{align*}
{\cal C} := \big\{ s \in {\cal L}_\phi: \phi(s) \in {\cal T}_n^c(\tau)\big\}.
\end{align*}
Then, we have
\begin{align}
P_{Y^n}({\cal T}_n^c(\tau)) &= \sum_{s \in {\cal C}} P_{X^{|s|}}(s) \nonumber \\
&= \sum_{s \in {\cal B} \cap {\cal C}}  P_{X^m}(s) + \sum_{s \in {\cal B}^c \cap {\cal C}}  P_{X^{|s|}}(s) \nonumber \\
&\le \sum_{s \in {\cal B} \cap {\cal C}}  P_{X^m}(s) + \sum_{s \in {\cal B}^c}  P_{X^{|s|}}(s) \nonumber \\
&= \sum_{s \in {\cal B} \cap {\cal C}}  P_{X^m}(s) + \Pr( T > m), \label{proof:single-shot-converse-1}
\end{align}
where the first identity follows from \eqref{eq:validity}.
Furthermore, we have
\begin{align}
\sum_{s \in {\cal B} \cap {\cal C}}  P_{X^m}(s) &= \sum_{s \in {\cal B} \cap {\cal C} \cap {\cal S}_m(\lambda)}  P_{X^m}(s) + \sum_{s \in {\cal B} \cap {\cal C} \cap {\cal S}_m^c(\lambda)}  P_{X^m}(s) \nonumber \\
&\le P_{X^m}({\cal S}_m(\lambda)) + \sum_{s \in {\cal B} \cap {\cal C} \cap {\cal S}_m^c(\lambda)}  P_{X^m}(s) \nonumber \\
&\le P_{X^m}({\cal S}_m(\lambda)) + \sum_{s \in {\cal B} \cap {\cal C} \cap {\cal S}_m^c(\lambda)} P_{Y^n}(\phi(s)) \nonumber \\
&\le P_{X^m}({\cal S}_m(\lambda)) + \sum_{s \in {\cal B} \cap {\cal C} \cap {\cal S}_m^c(\lambda)} 2^{-\tau} \nonumber \\
&\le P_{X^m}({\cal S}_m(\lambda)) + 2^{-\tau+\lambda}, \label{proof:single-shot-converse-2}
\end{align}
where \textchange{the second inequality follows from $\phi(s) \in {\cal Y}^n$ for  $s \in {\cal C}$,}
the third inequality follows from $\phi(s) \in {\cal T}^c_n(\tau)$ for $s \in {\cal C}$,
and the last inequality follows from the bound $|{\cal S}_m^c(\lambda)| \le 2^\lambda$. By combining \eqref{proof:single-shot-converse-1} and \eqref{proof:single-shot-converse-2},
we obtain \eqref{eq:single-shot-converse-1}; then, \eqref{eq:single-shot-converse-2} follows from \eqref{eq:single-shot-converse-1}.
\end{proof}

\section{Performance of Interval Algorithm} \label{sec:performance-interval}

First, we review a sequential version of the interval algorithm.\footnote{Unlike the interval algorithm in \cite{HanHoshi97}, we output each symbol of
the target process sequentially; however, there is no difference in performance analyses.}
In the algorithm, we sequentially update intervals 
\begin{align*}
{\cal I}_s &:= [\underline{\alpha}_s,\overline{\alpha}_s), \\
{\cal J}_t &:= [\underline{\beta}_t,\overline{\beta}_t)
\end{align*}
induced by coin process and target process, respectively.
For the null sequence $s=t=\bot$, we initially set $\underline{\alpha}_s = \underline{\beta}_t = 0$ and $\overline{\alpha}_s = \overline{\beta}_t=1$.
For a given sequence $s \in {\cal X}^i$ and $x \in {\cal X}$, the interval of coin process is updated by
\begin{align*}
\underline{\alpha}_{sx} &:= \underline{\alpha}_s + ( \overline{\alpha}_s - \underline{\alpha}_s) P_{(x-1)|s}, \\
\overline{\alpha}_{sx} &:= \underline{\alpha}_s + ( \overline{\alpha}_s - \underline{\alpha}_s) P_{x|s},
\end{align*}
where 
\begin{align*}
P_{x|s} := \sum_{k=1}^x P_{X_{i+1}|X^i}(k|s)
\end{align*}
for $x \in {\cal X}$ and $P_{0|s}=0$. Similarly, for a given sequence $t \in {\cal Y}^j$ and $y \in {\cal Y}$, the interval of target process is updated by
\begin{align*}
\underline{\beta}_{ty} &:= \underline{\beta}_t + (\overline{\beta}_t - \underline{\beta}_t) Q_{(y-1)|t}, \\
\overline{\beta}_{ty} &:= \underline{\beta}_t + (\overline{\beta}_t - \underline{\beta}_t) Q_{y|t},
\end{align*}
where 
\begin{align*}
Q_{y|t} := \sum_{k=1}^y P_{Y_{j+1}|Y^j}(k|t)
\end{align*}
for $y \in {\cal Y}$ and $Q_{0|t}=0$. Using these intervals, the algorithm proceeds as follows:
\begin{enumerate}
\item Set $s=t=\bot$, $i=0$, and $j=1$.

\item \label{IntAlg-step2} If $[\underline{\alpha}_s, \overline{\alpha}_s) \subseteq [\underline{\beta}_{ty}, \overline{\beta}_{ty})$ for some $y \in {\cal Y}$, 
then output $y_j = y$ and go to Step \ref{IntAlg-step3}; otherwise, set $i=i+1$, $s = s x_i$, and repeat Step \ref{IntAlg-step2} again. 

\item \label{IntAlg-step3} If $j=n$, terminates; otherwise, set $t=ty_j$, $j=j+1$, and go to Step \ref{IntAlg-step2}.
\end{enumerate}

The following example illustrates a behavior of the interval algorithm for converting
a Markov process to an i.i.d. process.
\begin{example} \label{example:behavior}
Let the coin process $\{ X^m \}_{m=1}^\infty$ be the Markov chain induced by the transition matrix in Fig.~\ref{Fig:Markov}
with the stationary initial distribution $P_{X_1}(1)=P_{X_1}(2)=1/2$; let $Y^2 = (Y_1,Y_2)$ be $2$ symbols of i.i.d. random variables
with $P_{Y_j}(1)=1/3$ and $P_{Y_j}(2)=2/3$ for $j=1,2$. In this case, updates of the intervals are described in Fig.~\ref{Fig:Markov-update-interval}.
Also, the algorithm tree is described in Fig.~\ref{Fig:Markov-tree}. For instance, when $X_1=2$ is observed, the algorithm outputs
$Y_1=2$; then, if $(X_2,X_3)=(1,2)$ are observed after $X_1=2$, the algorithm outputs $Y_2=2$ and terminates. On the other hand, when 
$(X_1,X_2)=(1,2)$ are observed, the algorithm first outputs $Y_1=2$; then, outputs $Y_2=1$ without further observing the coin process. 
In the latter case, the node in the algorithm tree is labeled by two symbols $(2,1)$.
\end{example}

\begin{figure}[t]
\centering{
\includegraphics[width=0.4\textwidth]{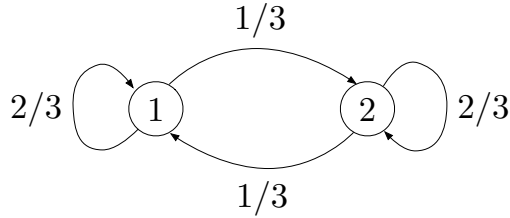}
\caption{A description of transition matrix in Example \ref{example:behavior}.}
\label{Fig:Markov}
}
\end{figure}

\begin{figure}[t]
\centering{
\includegraphics[width=0.5\textwidth]{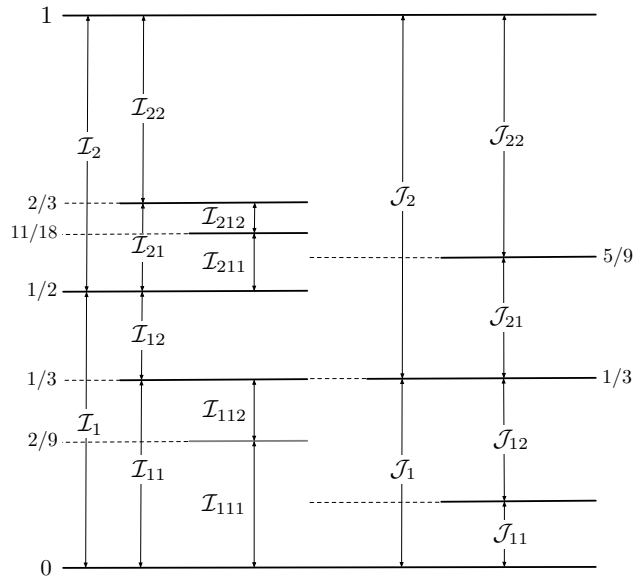}
\caption{A description of interval partitioning in Example \ref{example:behavior}.}
\label{Fig:Markov-update-interval}
}
\end{figure}

\begin{figure}[t]
\centering{
\includegraphics[width=0.6\textwidth]{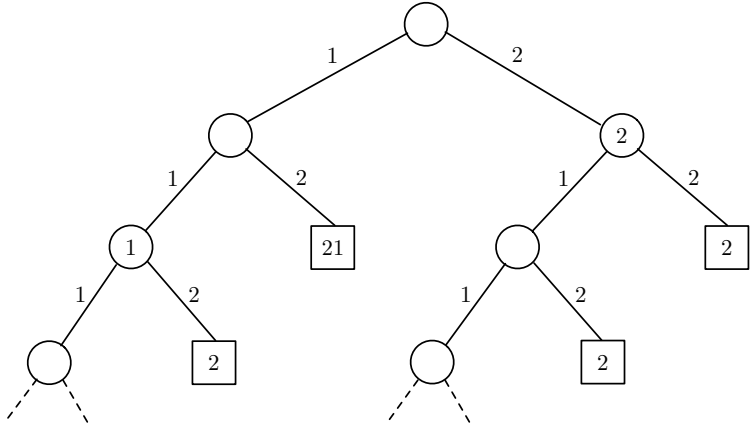}
\caption{A description of algorithm tree in Example \ref{example:behavior}. The label $1$ or $2$ 
of each edge indicates the outcome of coin random variable $X_i$. The label of each node indicates 
that that symbol(s) are output after observing the coin process up to that node. The square node indicates that
the algorithm terminates at the node.}
\label{Fig:Markov-tree}
}
\end{figure}

For notational convenience, we denote the function corresponding to the interval algorithm (cf.~\eqref{eq:functional-algorithm}) by $\phi_{\mathtt{int}}(\cdot)$.
Before verifying the validity of the algorithm (cf.~\eqref{eq:validity}) carefully,
we first examine the stopping time of the interval algorithm.
\begin{theorem} \label{theorem:performance-interval}
For the interval algorithm, the overflow probability of the stopping time satisfies
\begin{align*}
\Pr( T > m) \le P_{X^m}({\cal S}_m^c(\lambda)) + P_{Y^n}({\cal T}_n^c(\tau)) + 2^{-\lambda+\tau+1},
\end{align*}
where ${\cal S}_m(\lambda)$ and ${\cal T}_n(\tau)$ are defined as in \eqref{eq:typical-X} and \eqref{eq:typical-Y}, respectively. 
\end{theorem}
\begin{proof}
Let 
\begin{align*}
{\cal D}_m := \big\{ x^m \in {\cal X}^m : \forall y^n \in {\cal Y}^n,~{\cal I}_{x^m} \not\subseteq {\cal J}_{y^n} \big\}
\end{align*}
and
\begin{align*}
{\cal E}_m := \big\{ x^m \in {\cal X}^m : \exists y^n \in {\cal T}_n(\tau) \mbox{ s.t. } {\cal I}_{x^m} \cap {\cal J}_{y^n} \neq \emptyset \big\}.
\end{align*}
Then, since the algorithm does not terminate after observing $x^m$ if and only if $x^m \in {\cal D}_m$,
the overflow probability can be rewritten as 
\begin{align}
\Pr( T > m) &= \sum_{x^m \in {\cal D}_m} P_{X^m}(x^m) \nonumber \\
&= \sum_{x^m \in {\cal D}_m \cap {\cal E}_m} P_{X^m}(x^m) + \sum_{x^m \in {\cal D}_m \cap {\cal E}_m^c} P_{X^m}(x^m) \nonumber \\
&\le \sum_{x^m \in {\cal D}_m \cap {\cal E}_m} P_{X^m}(x^m) + P_{Y^n}({\cal T}_n^c(\tau)), \label{eq:proof-int-bound-1}
\end{align}
where the inequality is justified as follows. Note that $x^m \in {\cal E}_m^c$ implies ${\cal I}_{x^m} \cap {\cal J}_{y^n} = \emptyset$ 
for every $y^n \in {\cal T}_n(\tau)$, which further implies 
\begin{align*}
\bigcup_{x^m \in {\cal E}_m^c} {\cal I}_{x^m} \subseteq \bigcup_{y^n \in {\cal T}^c_n(\tau)} {\cal J}_{y^n}.
\end{align*}
\textchange{Thus, by noting that $P_{X^m}(x^m) = |{\cal I}_{x^m}|$ and $P_{Y^n}(y^n) = |{\cal J}_{y^n}|$, we have the inequality.} 

Furthermore, the first term of \eqref{eq:proof-int-bound-1} can be bounded as
\begin{align}
\sum_{x^m \in {\cal D}_m \cap {\cal E}_m} P_{X^m}(x^m)
&= \sum_{x^m \in  {\cal D}_m \cap {\cal E}_m \cap {\cal S}_m(\lambda)} P_{X^m}(x^m) + \sum_{x^m \in  {\cal D}_m \cap {\cal E}_m \cap {\cal S}_m^c(\lambda)} P_{X^m}(x^m) \nonumber \\
&\le \sum_{x^m \in  {\cal D}_m \cap {\cal E}_m \cap {\cal S}_m(\lambda)} P_{X^m}(x^m) + P_{X^m}({\cal S}_m^c(\lambda)) \nonumber \\
&\le \sum_{x^m \in  {\cal D}_m \cap {\cal E}_m \cap {\cal S}_m(\lambda)} 2^{-\lambda} + P_{X^m}({\cal S}_m^c(\lambda)) \nonumber \\
&\le |{\cal D}_m \cap {\cal E}_m| 2^{-\lambda} + P_{X^m}({\cal S}_m^c(\lambda)) \nonumber \\
&\le 2 |{\cal T}_n(\tau)| 2^{-\lambda} + P_{X^m}({\cal S}_m^c(\lambda)) \nonumber \\ 
&\le 2^{-\lambda+\tau+1} + P_{X^m}({\cal S}_m^c(\lambda)), \label{eq:proof-int-bound-2}
\end{align}
where the second last inequality is justified as follows. 
\textchange{
By noting that $x^m \in {\cal D}_m \cap {\cal E}_m$ implies ${\cal I}_{x^m} \cap {\cal J}_{y^n} \neq \emptyset$ 
and ${\cal I}_{x^m} \not\subseteq {\cal J}_{y^n}$ for some $y^n \in {\cal T}_n(\tau)$, we have
\begin{align*}
|{\cal D}_m \cap {\cal E}_m | \le \sum_{y^n \in {\cal T}_n(\tau)} \sum_{x^m \in {\cal X}^m} \bol{1}\big[ {\cal I}_{x^m} \cap {\cal J}_{y^n} \neq \emptyset,  {\cal I}_{x^m} \not\subseteq {\cal J}_{y^n} \big].
\end{align*}
For each fixed $y^n \in {\cal T}_n(\tau)$, if there are more than two $x^m$'s satisfying ${\cal I}_{x^m} \cap {\cal J}_{y^n} \neq \emptyset$,
then all but the top and bottom ones must satisfy ${\cal I}_{x^m} \subseteq {\cal J}_{y^n}$; in other words, 
there are at most two $x^m$'s satisfying both the conditions in the indicator function. Thus, we have
\begin{align*}
\sum_{y^n \in {\cal T}_n(\tau)} \sum_{x^m \in {\cal X}^m} \bol{1}\big[ {\cal I}_{x^m} \cap {\cal J}_{y^n} \neq \emptyset,  {\cal I}_{x^m} \not\subseteq {\cal J}_{y^n} \big] \le 2 |{\cal T}_n(\tau)|.
\end{align*}
}

Finally, by combining \eqref{eq:proof-int-bound-1} and \eqref{eq:proof-int-bound-2}, we have the claimed bound. 
\end{proof}

Now, we argue the validity of the interval algorithm. Clearly, if the coin process is deterministic, the random 
number generation is not possible. By using Theorem \ref{theorem:performance-interval},
we can prove that the interval algorithm exactly generate a target distribution
as long as the coin process has ``diverging randomness". 
\begin{corollary} \label{corollary-validity}
If the coin process $\bm{X}=\{ X^m\}_{m=1}^\infty$ satisfies 
\begin{align} \label{eq:non-trivial-randomness}
\lim_{m\to\infty} P_{X^m}({\cal S}_m^c(\lambda)) = 0
\end{align}
for every $\lambda > 0$, where ${\cal S}_m(\lambda)$ is defined as in \eqref{eq:typical-X}, then the interval algorithm is valid, i.e., 
\begin{align*}
\lim_{m\to\infty} \Pr( \phi_{\mathtt{int}}(X^m) = y^n) = P_{Y^n}(y^n)
\end{align*}
for every $y^n \in {\cal Y}^n$. 
\end{corollary}
\begin{proof}
Upon observing $x^m \in {\cal X}^m$, the interval algorithm terminates with output $y^n \in {\cal Y}^n$ if and only if
${\cal I}_{x^m} \subseteq {\cal J}_{y^n}$. Thus, we have
\begin{align} \label{eq:validity-inequality}
\Pr( \phi_{\mathtt{int}}(X^m) = y^n) \le P_{Y^n}(y^n).
\end{align}
Furthermore, since the lefthand side of \eqref{eq:validity-inequality} is non-decreasing in $m$, it has a limit.
To prove that the limit coincides with the righthand side, we apply Theorem  \ref{theorem:performance-interval} with
\begin{align*}
\tau = \max_{y^n \in \mathtt{supp}(P_{Y^n})} \log \frac{1}{P_{Y^n}(y^n)},
\end{align*}
where $\mathtt{supp}(P_{Y^n})$ is the support of distribution $P_{Y^n}$.
Then, we have
\begin{align*}
\Pr(\phi_{\mathtt{int}}(X^m) \notin {\cal Y}^n) &= \Pr(T > m) \\
&\le P_{X^m}({\cal S}_m^c(\lambda)) + 2^{-\lambda + \tau + 1}
\end{align*}
for any $\lambda$. Since \eqref{eq:non-trivial-randomness} holds for any $\lambda$ by assumption,
by taking the limit $m\to \infty$ and $\lambda \to \infty$ with the diagonalization argument (cf.~\cite{han:book}), 
we have
\begin{align*}
\lim_{m\to\infty} \Pr(\phi_{\mathtt{int}}(X^m) \in {\cal Y}^n) = 1,
\end{align*}
which together with \eqref{eq:validity-inequality} implies the claim of the theorem.\footnote{\textchange{Note that
$a \le A$, $b \le B$, $A+B=1$, and $a+b=1$ imply $1-b = a \le A = 1-B$, i.e., $B \le b$.}} 
\end{proof}

In fact, using Theorem \ref{theorem:converse},
we can also prove that the same condition as Corollary \ref{corollary-validity} is necessary for exact random number generation by any algorithms.
\begin{corollary}
If the coin process $\bm{X}=\{ X^m\}_{m=1}^\infty$ does not satisfy \eqref{eq:non-trivial-randomness} for some $\lambda>0$,
then there exists a target distribution $P_{Y^n}$ with sufficiently large $n$ such that 
the validity \eqref{eq:validity} does not hold for any random number generation algorithms.
\end{corollary}
\begin{proof}
Suppose that \eqref{eq:non-trivial-randomness} does not hold for some $\lambda>0$, i.e., there exists $\delta > 0$ such that 
\begin{align*}
P_{X^m}({\cal S}_m^c(\lambda)) \ge \delta
\end{align*}
for every sufficiently large $m$. Let $\tau = \lambda- \log(\delta/2)$. Then, by \eqref{eq:single-shot-converse-2} of Theorem \ref{theorem:converse}, we have
\begin{align*}
\Pr( T > m) \ge \frac{\delta}{2} - P_{Y^n}({\cal T}_n(\tau)).
\end{align*}
This bound implies that, for any target distribution with min-entropy $H_{\min}(Y^n) > \tau$, 
\begin{align*}
\Pr( T > m) \ge \frac{\delta}{2}
\end{align*}
for every sufficiently large $m$.
Thus, the validity \eqref{eq:validity} cannot be satisfied for some $y^n \in {\cal Y}^n$.
\end{proof}

For instance, any absorbing Markov chain does not satisfy the sufficient condition of Corollary \ref{corollary-validity}; note that
absorbing Markov chains have $0$ spectral inf-entropy, i.e., $\underline{H}(\bm{X})=0$.
A further relaxed sufficient condition is that $\underline{H}(\bm{X}) > 0$; however, this relaxed condition is not necessary
in general as the following example illustrates.
\begin{example}[Harmonic Bernoulli Coin]
Let us consider independent but non-stationary Bernoulli trials $\bm{X}= \{ X^m \}_{m=1}^\infty$
such that $P_{X_i}(1) = 2^{-1/i}$. Then, since \textchange{the min-entropy (see the last paragraph of Section \ref{section:introduction} for the definition) of $X^m$ is bounded from below as}
\begin{align*}
H_{\min}(X^m) &= \sum_{i=1}^m H_{\min}(X_i) \\
&= \sum_{i=1}^m \frac{1}{i} \\
&\ge \ln (m+1),
\end{align*}
\eqref{eq:non-trivial-randomness} is satisfied for any $\lambda >0$. Thus, this coin process can be used for
the interval algorithm. However, we can verify that $\underline{H}(\bm{X})=0$ as follows. Note that
\begin{align*}
\frac{1}{m} H(X^m) &= \sum_{i=1}^m \frac{1}{m} H(X_i) \\
&\le h\bigg( \sum_{i=1}^m \frac{1}{m} P_{X_i}(1) \bigg)
\end{align*}
by concavity of the entropy. Furthermore, since $t \mapsto 2^{-t}$ is convex, we have
\begin{align*}
\sum_{i=1}^m \frac{1}{m} P_{X_i}(1) 
&= \sum_{i=1}^m \frac{1}{m} 2^{- H_{\min}(X_i)} \\
&\ge 2^{- \sum_{i=1}^m \frac{1}{m} H_{\min}(X_i)} \\
&\ge 2^{- \frac{\ln m + 1}{m}} \\
&\ge \frac{1}{2},
\end{align*}
which implies 
\begin{align*}
\frac{1}{m} H(X^m) \le h\big( 2^{- \frac{\ln m + 1}{m}} \big).
\end{align*}
Thus, by \cite[Theorem 1.7.2]{han:book}, we have
\begin{align*}
\underline{H}(\bm{X}) \le \lim_{m\to\infty} \frac{1}{m} H(X^m) = 0.
\end{align*}

On the other hand, if the probability distribution of each trial is $P_{X_i}(1)=2^{-1/i^2}$, then, for any $\delta > 0$, we have
\begin{align*}
\lim_{m\to\infty} P_{X^m}( {\cal S}_m^c(\pi^2/6 + \delta)) \ge \prod_{i=1}^\infty P_{X_i}(1) = 2^{- \pi^2/6}. 
\end{align*}
Thus, this coin process cannot be used for any random number generation algorithms.
\end{example}

\section{Asymptotic Analysis} \label{sec:asymptotic}

\subsection{General Results}

In this section, we shall examine the asymptotic optimality of the interval algorithm. 
Recall the notations of information measures described in \eqref{eq:notation-sup-entropy}, \eqref{eq:notation-inf-entropy},
and \eqref{eq:notation-entropy}. 
We start with the criterion of the overflow probability of the stopping time. 
\begin{definition}
For a given random number generation algorithm converting $\bm{X}$ to $\bm{Y}$, a rate $R \ge 0$ is defined to be achievable if the stopping time $T_n$ satisfies 
\begin{align*}
\lim_{n\to\infty} \Pr( T_n > nR) = 0.
\end{align*}
Then, let $R^\star_{\mathtt{int}}(\bm{X},\bm{Y})$ and $R^\star(\bm{X},\bm{Y})$ be the infimum rates that are achievable by
the interval algorithm and by any algorithm (not necessarily the interval algorithm), respectively. 
\end{definition}

\begin{theorem} \label{theorem:stopping-interval-asymptotic}
For given coin process $\bm{X}$ with $\underline{H}(\bm{X})>0$ and target process $\bm{Y}$, the infimum achievable rate of the interval algorithm satisfies 
\begin{align} \label{eq:theorem-stopping-interval}
R^\star_{\mathtt{int}}(\bm{X},\bm{Y}) \le  \frac{\overline{H}(\bm{Y})}{\underline{H}(\bm{X})}.
\end{align}
On the other hand, the infimum achievable rate of any algorithm satisfies 
\begin{align} \label{eq:theorem-stopping-any}
R^\star(\bm{X},\bm{Y}) \ge \max\bigg[ \frac{\overline{H}(\bm{Y})}{\overline{H}(\bm{X})}, \frac{\underline{H}(\bm{Y})}{\underline{H}(\bm{X})} \bigg]. 
\end{align}
\end{theorem}
\begin{proof}
We first prove \eqref{eq:theorem-stopping-interval}. Fix arbitrary $\delta_1,\delta_2>0$. By applying Theorem \ref{theorem:performance-interval} with
\begin{align*}
m_n &= n \bigg( \frac{\overline{H}(\bm{Y})}{\underline{H}(\bm{X})} + \delta_2 \bigg), \\
R &= \frac{\overline{H}(\bm{Y})}{\underline{H}(\bm{X})} + \delta_2, \\
\lambda &= m_n ( \underline{H}(\bm{X}) - \delta_1), \\
\tau &= n( \overline{H}(\bm{Y}) + \delta_1),
\end{align*}
we can bound the overflow probability of the stopping time for the interval algorithm as
\begin{align*}
\Pr( T_n > n R) &\le \Pr\bigg( \frac{1}{m_n} \log \frac{1}{P_{X^{m_n}}(X^{m_n})} < \underline{H}(\bm{X}) - \delta_1 \bigg)
+ \Pr\bigg( \frac{1}{n} \log \frac{1}{P_{Y^n}(Y^n)} > \overline{H}(\bm{Y}) + \delta_1 \bigg) \\
&~~~ + \exp\bigg[ - n \bigg\{ \delta_2( \underline{H}(\bm{X}) - \delta_1) - \delta_1\bigg( \frac{\overline{H}(\bm{Y})}{\underline{H}(\bm{X})} + 1 \bigg) \bigg\} + 1 \bigg].
\end{align*}
Thus, if we take $\delta_1$ sufficiently small compared to $\delta_2$, we have
\begin{align*}
\lim_{n\to\infty} \Pr( T_n > n R) = 0,
\end{align*}
which implies that $R=\overline{H}(\bm{Y})/\underline{H}(\bm{X})+\delta_2$ is achievable. Since $\delta_2>0$ is arbitrary, we have \eqref{eq:theorem-stopping-interval}.

Next, we prove the first bound of \eqref{eq:theorem-stopping-any}. 
Fix arbitrary $\delta_1,\delta_2>0$. By applying \eqref{eq:single-shot-converse-1} of Theorem \ref{theorem:converse} with
\begin{align*}
m_n &= n \bigg( \frac{\overline{H}(\bm{Y})}{\overline{H}(\bm{X})} - \delta_2 \bigg), \\
R &= \frac{\overline{H}(\bm{Y})}{\overline{H}(\bm{X})} - \delta_2, \\
\lambda &= m_n ( \overline{H}(\bm{X}) + \delta_1), \\
\tau &= n( \overline{H}(\bm{Y}) - \delta_1),
\end{align*}
for any random number generation algorithms, we have
\begin{align*}
\Pr( T_n > n R) &\ge \Pr\bigg( \frac{1}{n} \log \frac{1}{P_{Y^n}(Y^n)} > \overline{H}(\bm{Y}) - \delta_1 \bigg)
 - \Pr\bigg( \frac{1}{m_n} \log \frac{1}{P_{X^{m_n}}(X^{m_n})} \ge \overline{H}(\bm{X}) + \delta_1 \bigg) \\
 &~~~- \exp\bigg[ - n \bigg\{ \delta_2 ( \overline{H}(\bm{X})+\delta_1)  - \delta_1 \bigg( \frac{\overline{H}(\bm{Y})}{\overline{H}(\bm{X})}+1 \bigg) \bigg\} \bigg].
\end{align*} 
Thus, if we take $\delta_1$ sufficiently small compared to $\delta_2$, the definition of $\overline{H}(\bm{Y})$ leads to
\begin{align*}
\liminf_{n\to\infty} \Pr(T_n > nR) > 0,
\end{align*}
which implies that $R = \overline{H}(\bm{Y})/ \overline{H}(\bm{X}) - \delta_2$ is not achievable. Since $\delta_2$ is arbitrary, we have
the first bound of \eqref{eq:theorem-stopping-any}. We can prove the second bound of \eqref{eq:theorem-stopping-any}
in a similar manner by using \eqref{eq:single-shot-converse-2} of Theorem \ref{theorem:converse}.
\end{proof}

When either the coin or the target process has one point spectrum, we immediately obtain the following corollary from Theorem \ref{theorem:stopping-interval-asymptotic}.

\begin{corollary} \label{corollary:asymptotic-stopping-overflow}
When the spectral sup-entropy $\overline{H}(\bm{X})$ and inf-entropy $\underline{H}(\bm{X})$ of coin process 
coincide with its  entropy rate $H(\bm{X})$,\footnote{When $\overline{H}(\bm{X})=\underline{H}(\bm{X})$, called the {\em one-point} spectrum, the limit in \eqref{eq:notation-entropy} exists,
and we have $\overline{H}(\bm{X})=\underline{H}(\bm{X}) = H(\bm{X})$ (cf.~\cite[Theorem 1.7.2]{han:book}).} we have
\begin{align*}
R^\star_{\mathtt{int}}(\bm{X},\bm{Y}) = R^\star(\bm{X},\bm{Y}) = \frac{\overline{H}(\bm{Y})}{H(\bm{X})}.
\end{align*}
On the other hand, when spectral sup-entropy $\overline{H}(\bm{Y})$ and inf-entropy $\underline{H}(\bm{Y})$ of
the target process coincide with its entropy rate $H(\bm{Y})$, we have
\begin{align*}
R^\star_{\mathtt{int}}(\bm{X},\bm{Y}) = R^\star(\bm{X},\bm{Y}) = \frac{H(\bm{Y})}{\underline{H}(\bm{X})}.
\end{align*}
\end{corollary}

Next, we investigate the average stopping time $\mathbb{E}[T_n]$.
\begin{definition}
For a given random number generation algorithm converting $\bm{X}$ to $\bm{Y}$, a rate $L \ge 0$
is defined to be average achievable if the average stopping time $\mathbb{E}[T_n]$ satisfies
\begin{align*}
\limsup_{n\to\infty} \frac{\mathbb{E}[T_n]}{n} \le L.
\end{align*}
Then, let $L_{\mathtt{int}}^\star(\bm{X},\bm{Y})$ and $L^\star(\bm{X},\bm{Y})$ be the infimum rates that are average achievable 
by the interval algorithm and by any algorithm (not necessarily the interval algorithm), respectively. 
\end{definition}

In the following argument, as a technical condition, we assume that the upper and lower tails of the information spectrum of the coin process vanish
sufficiently rapidly in the following sense:  
for any $\delta>0$, there exist constants $K$ and $m_0=m_0(\delta,K)$ such that
\begin{align} \label{eq:upper-tail}
\Pr\bigg( \frac{1}{m}\log \frac{1}{P_{X^m}(X^m)} \ge \overline{H}(\bm{X})+\delta \bigg) \le \frac{K}{m^2}
\end{align}
and
\begin{align} \label{eq:lower-tail}
\Pr\bigg( \frac{1}{m}\log \frac{1}{P_{X^m}(X^m)} \le \underline{H}(\bm{X})-\delta \bigg) \le \frac{K}{m^2}
\end{align}
for every $m \ge m_0$. In fact, i.i.d. processes, irreducible Markov processes, or mixture of those processes
satisfy much stronger requirement, i.e., the upper and lower tails vanish exponentially \cite{dembo-zeitouni-book}.

Now, we are ready to present the asymptotic behavior of the average stopping time.
\begin{theorem} \label{theorem:asymptotic-average}
For given coin process $\bm{X}$ satisfying \eqref{eq:lower-tail} and target process $\bm{Y}$, the infimum average achievable rate of the interval algorithm satisfies 
\begin{align} \label{eq:asymptotic-average-interval-algorithm}
L_{\mathtt{int}}^\star(\bm{X}, \bm{Y}) \le \frac{H(\bm{Y})}{\underline{H}(\bm{X})}.
\end{align}
On the other hand, 
for given coin process $\bm{X}$ satisfying \eqref{eq:upper-tail} and target process $\bm{Y}$, 
the infimum average achievable rate of any algorithm satisfies 
\begin{align}  \label{eq:asymptotic-average-any-algorithm}
L^\star(\bm{X},\bm{Y}) \ge \frac{H(\bm{Y})}{\overline{H}(\bm{X})}.
\end{align}
\end{theorem}
\begin{proof}
We first prove \eqref{eq:asymptotic-average-interval-algorithm}. By using the identity on the expectation (eg.~see \cite[Eq.~(21.9)]{billingsley-book}), we can write
\begin{align} \label{eq:proof-asymptotic-average-interval}
\mathbb{E}[T_n] = \int_0^\infty \Pr( T_n > z) dz.
\end{align}
Fix arbitrary $\delta > 0$. For each $z$, by applying Theorem \ref{theorem:performance-interval} with $\lambda = \lfloor z \rfloor ( \underline{H}(\bm{X}) - \delta)$
and $\tau = z ( \underline{H}(\bm{X})-2\delta)$, we have
\begin{align*}
\Pr( T_n > z) &\le \Pr( T_n > \lfloor z \rfloor ) \\
&\le \Pr\bigg( \frac{1}{\lfloor z \rfloor} \log \frac{1}{P_{X^{\lfloor z \rfloor}}(X^{\lfloor z \rfloor})} < \underline{H}(\bm{X})-\delta \bigg) + \Pr\bigg( \frac{1}{(\underline{H}(\bm{X})-2\delta)} \log \frac{1}{P_{Y^n}(Y^n)} > z \bigg)
 + 2^{-\delta z+1}.
\end{align*}
The integral of the first term is bounded as 
\begin{align*}
\int_0^\infty  \Pr\bigg( \frac{1}{\lfloor z \rfloor} \log \frac{1}{P_{X^{\lfloor z \rfloor}}(X^{\lfloor z \rfloor})} < \underline{H}(\bm{X})-\delta \bigg) dz 
&\le m_0 + \int_{m_0}^\infty \frac{K}{(z-1)^2} dz \\
&= m_0 + \frac{K}{(m_0-1)},
\end{align*}
where the inequality follows from \eqref{eq:lower-tail}; the integral of the second term is 
\begin{align*}
\int_0^\infty  \Pr\bigg( \frac{1}{(\underline{H}(\bm{X})-2\delta)} \log \frac{1}{P_{Y^n}(Y^n)} > z \bigg) dz = \frac{H(Y^n)}{(\underline{H}(\bm{X})-2\delta)},
\end{align*}
where we used the identity on the expectation again; the integral of the third term is given by $\frac{2}{\delta \ln 2}$. 
By substituting these evaluations into \eqref{eq:proof-asymptotic-average-interval}, we obtain
\begin{align*}
\limsup_{n\to\infty} \frac{\mathbb{E}[T_n]}{n} \le \frac{H(\bm{Y})}{(\underline{H}(\bm{X})-2\delta)}.
\end{align*} 
Since $\delta > 0$ is arbitrary, we have \eqref{eq:asymptotic-average-interval-algorithm}.

Next, we prove \eqref{eq:asymptotic-average-any-algorithm}. We start with \eqref{eq:proof-asymptotic-average-interval}. Fix arbitrary $\delta > 0$.
For each $z$, by applying \eqref{eq:single-shot-converse-1} of Theorem \ref{theorem:converse} with $\lambda = \lceil z \rceil (\overline{H}(\bm{X})+\delta)$
and $\tau = z(\overline{H}(\bm{X})+2\delta)$, we have
\begin{align*}
\lefteqn{ \Pr( T_n > z) } \\
&\ge \Pr(T_n > \lceil z \rceil ) \\
&\ge \Pr\bigg( \frac{1}{(\overline{H}(\bm{X})+2\delta)} \log \frac{1}{P_{Y^n}(Y^n)} > z \bigg)
 - \Pr\bigg( \frac{1}{\lceil z \rceil} \log \frac{1}{P_{X^{\lceil z \rceil}}(X^{\lceil z \rceil})} \ge \overline{H}(\bm{X}) + \delta \bigg)
 - 2^{-\delta (z-1) + \overline{H}(\bm{X})}.
\end{align*}
By evaluating the integral of each term in a similar manner as above and by substituting them into \eqref{eq:proof-asymptotic-average-interval}, we obtain
\begin{align*}
\limsup_{n\to\infty} \frac{\mathbb{E}[T_n]}{n} \ge \frac{H(\bm{Y})}{(\overline{H}(\bm{X})+2\delta)}.
\end{align*}
Since $\delta > 0$ is arbitrary, we have \eqref{eq:asymptotic-average-any-algorithm}.
\end{proof}

When the coin process has one point spectrum, we immediately obtain the following corollary from Theorem \ref{theorem:asymptotic-average}. 
\begin{corollary} \label{corollary:average-asymptotic}
When the coin process satisfies \eqref{eq:upper-tail}, \eqref{eq:lower-tail}, and $\overline{H}(\bm{X})$ and  $\underline{H}(\bm{X})$ coincide with $H(\bm{X})$, we have
\begin{align*}
L_{\mathtt{int}}^\star(\bm{X}, \bm{Y}) = L^\star(\bm{X},\bm{Y}) = \frac{H(\bm{Y})}{H(\bm{X})}.
\end{align*}
\end{corollary}

It should be noted that the target process $\bm{Y}$ need not to have one point spectrum in Corollary \ref{corollary:average-asymptotic}.

\begin{remark} \label{remark:implication-almost-sure}
When both the coin and target processes are ergodic, it was shown in \cite{UyeKan99} that 
the normalized stopping time of the interval algorithm almost surely converges to the ratio of the entropy rates, i.e.,
\begin{align} \label{eq:almost-sure}
\lim_{n\to\infty} \frac{T_n}{n} = \frac{H(\bm{Y})}{H(\bm{X})}~~~\mbox{a.s.}
\end{align}
This result immediately implies 
\begin{align*}
R_\mathtt{int}^\star(\bm{X},\bm{Y}) = \frac{H(\bm{Y})}{H(\bm{X})}.
\end{align*}
However, in order to derive 
\begin{align*}
L_\mathtt{int}^\star(\bm{X},\bm{Y}) = \frac{H(\bm{Y})}{H(\bm{X})}
\end{align*}
from \eqref{eq:almost-sure}, we need to prove uniform integrability of $T_n/n$ (cf.~\cite[Theorem 16.14]{billingsley-book}),
which is a cumbersome problem thought it may be possible.
\end{remark}


In the next subsections, we shall illustrate the general results above with concrete classes of coin/target processes.

\subsection{Markov Coin/Target Processes} \label{subsec:markov-coin-target}

As a coin process, we consider a Markov chain $\bm{X} = \{X^m \}_{m=1}^\infty$ on ${\cal X}$ induced by a transition matrix
$W(x|x^\prime)$. Suppose that $W$ is irreducible, i.e., for any $x,x^\prime \in {\cal X}$, there exists an integer $k\ge1$ such that $W^k(x|x^\prime)>0$.
When $W$ is irreducible, there exists a unique stationary distribution $\pi$ \cite{Meyer:book}. For the stationary distribution, let 
\begin{align*}
H^W(X) := \sum_{x,x^\prime} \pi(x^\prime) W(x|x^\prime) \log \frac{1}{W(x|x^\prime)},
\end{align*}
which coincides with the entropy rate of the Markov chain when the initial distribution is $\pi$ \cite{cover}.

We use the following bounds from the large deviation theory (cf.~\cite{dembo-zeitouni-book, WatHay:17}).
\begin{lemma} \label{lemma:large-deviation}
Let $\bm{X}=\{X^m\}_{m=1}^\infty$ be a Markov chain  induced by an irreducible transition matrix $W$ 
with arbitrary initial distribution $P_{X_1}$. For $\delta >0$, there exist $\overline{E}(\delta),\underline{E}(\delta)>0$ such that
\begin{align*}
\Pr\bigg( \frac{1}{m} \sum_{i=2}^m \log \frac{1}{W(X_i|X_{i-1})} \ge H^W(X) + \delta \bigg) &\le 2^{-m \overline{E}(\delta)}, \\
\Pr\bigg( \frac{1}{m} \sum_{i=2}^m \log \frac{1}{W(X_i|X_{i-1})} \le H^W(X) - \delta \bigg) &\le 2^{-m \underline{E}(\delta)} \\
\end{align*} 
for every sufficiently large $m$.
\end{lemma} 
 
From Lemma \ref{lemma:large-deviation} and \cite[Theorem 1.7.2]{han:book}, for any initial distribution $P_{X_1}$, we have
\begin{align} \label{eq:irreducible-case-Markov}
\underline{H}(\bm{X}) = H(\bm{X}) = \overline{H}(\bm{X}) = H^W(X).
\end{align} 
Furthermore, the conditions in \eqref{eq:upper-tail} and \eqref{eq:lower-tail} are also satisfied.
 
For the target process, we consider a Markov chain $\bm{Y}=\{ Y^n \}_{n=1}^\infty$ on ${\cal Y}$ induced by a transition matrix $V$.
Suppose that $V$ is not irreducible but there is no transient class (cf.~\cite{Meyer:book}), i.e., the transition matrix can be decomposed as a direct sum form:
\begin{align*}
V = \bigoplus_{\xi=1}^r V_\xi,
\end{align*} 
where $V_\xi$ is the irreducible transition matrix on irreducible class ${\cal Y}_\xi \subset {\cal Y}$ for $\xi=1,\ldots,r$.
When the initial state is $Y_1 \in {\cal Y}_\xi$, then $Y_2, Y_3,\ldots$ remain in the same irreducible class ${\cal Y}_\xi$.
Thus, for the weight 
\begin{align*}
w(\xi) = \Pr\big( Y_1 \in {\cal Y}_\xi \big)
\end{align*}
of each irreducible class induced from the initial distribution $P_{Y_1}$, the Markov chain $Y^n$ can be regarded as 
a mixture of irreducible Markov chains, i.e.,
\begin{align*}
\Pr\big( Y^n = y^n \big) = \sum_{\xi=1}^r w(\xi) \Pr\big( Y^n = y^n | Y_1 \in {\cal Y}_\xi \big).
\end{align*}
Let $\pi_\xi$ be the stationary distribution of $V_\xi$, and let 
\begin{align*}
H^{V_\xi}(Y) := \sum_{y,y^\prime \in {\cal Y}_\xi} \pi_\xi(y^\prime) V_\xi(y|y^\prime) \log \frac{1}{V_\xi(y|y^\prime)}
\end{align*}
be the entropy rate of $\xi$-th irreducible class.
Then, by the argument in \cite[Sec.~1.4]{han:book}, the information spectral quantities and the entropy 
rate are given as follows (see also Fig.~\ref{Fig:Markov-spectrum}):\footnote{When transition matrix $V$ has transient class, the 
information spectral quantities and the entropy rate are given by the same formulae; however, weight $w(\xi)$ is determined as
the limiting probability such that the initial state is eventually absorbed into irreducible class ${\cal Y}_\xi$ (cf.~\cite[Chapter 8]{Meyer:book}).}
\begin{align*}
\overline{H}(\bm{Y}) &= \max\big\{ H^{V_\xi}(Y) : 1 \le \xi \le r, w(\xi) > 0 \big\}, \\
\underline{H}(\bm{Y}) &= \min\big\{ H^{V_\xi}(Y) : 1 \le \xi \le r, w(\xi) > 0 \big\}, \\
H(\bm{Y}) &= \sum_{\xi=1}^r w(\xi) H^{V_\xi}(Y).
\end{align*}

\begin{figure}[t]
\centering{
\includegraphics[width=0.4\textwidth]{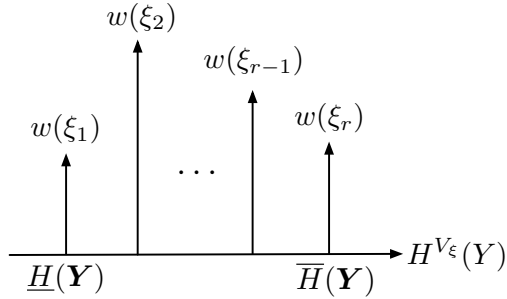}
\caption{Information spectrum of reducible Markov chain.}
\label{Fig:Markov-spectrum}
}
\end{figure}

From the above arguments along with Corollary \ref{corollary:asymptotic-stopping-overflow} and Corollary \ref{corollary:average-asymptotic}, 
we have 
\begin{align*}
R_{\mathtt{int}}^\star(\bm{X}, \bm{Y}) = R^\star(\bm{X},\bm{Y}) = \frac{1}{H^W(X)} \max\big\{ H^{V_\xi}(Y) : 1 \le \xi \le r, w(\xi) > 0 \big\}
\end{align*}
and
\begin{align*}
L_{\mathtt{int}}^\star(\bm{X}, \bm{Y}) = L^\star(\bm{X},\bm{Y}) = \frac{1}{H^W(X)} \sum_{\xi=1}^r w(\xi) H^{V_\xi}(Y).
\end{align*}

\subsection{Target Process with Continuous Spectrum}

As a coin process, we again consider a Markov chain $\bm{X}=\{X^m\}_{m=1}^\infty$ on ${\cal X}$ 
induced by an irreducible transition matrix $W$. 
As we have seen in Section \ref{subsec:markov-coin-target}, 
the spectral sup-entropy and inf-entropy coincide with the entropy rate, and
they are given by $H^W(X)$.

Let $\{ V_\xi \}_{\xi \in \Xi}$ be a parametrized family of irreducible matrix on ${\cal Y}$,
and let
\begin{align*}
P_{Y^n}(y^n) = \int P_{Y_\xi}^n(y^n) dw(\xi)
\end{align*} 
be the mixture of Markov process with arbitrary weight $w(\xi)$, where 
\begin{align*}
P_{Y_\xi^n}(y^n) = P_{Y_{\xi,1}}(y_1) \prod_{i=2}^n V_{\xi}(y_i | y_{i-1}).
\end{align*}
Then, for the target process $\bm{Y}=\{Y^n\}_{n=1}^\infty$, we have\footnote{For a measurable function
$Z_\xi$ of $\xi$, the {\em essential supremum} with respect to $w(\xi)$ is defined as 
$w\mbox{-ess.~sup } Z_\xi := \inf\{ \alpha : \Pr\{ Z_\xi > \alpha \} = 0$.} (cf.~\cite[Theorem 1.4.3]{han:book})
\begin{align*}
\overline{H}(\bm{Y})= w\mbox{-ess.~sup } H^{V_\xi}(Y)
\end{align*}
and (cf.~\cite[Remark 1.7.3]{han:book})
\begin{align*}
H(\bm{Y})= \int H^{V_\xi}(Y) dw(\xi).
\end{align*}

For the above described coin process and target process, Corollary \ref{corollary:asymptotic-stopping-overflow}
and Corollary \ref{corollary:average-asymptotic} immediately provide 
\begin{align*}
R^\star_{\mathtt{int}}(\bm{X},\bm{Y}) = R^\star(\bm{X},\bm{Y}) = \frac{1}{H^{W}(X)} w\mbox{-ess.~sup } H^{V_\xi}(Y)
\end{align*}
and
\begin{align*}
L^\star_{\mathtt{int}}(\bm{X},\bm{Y}) = L^\star(\bm{X},\bm{Y}) =  \frac{1}{H^{W}(X)} \int H^{V_\xi}(Y) dw(\xi).
\end{align*}

\section{Connection to Fixed-Length Random Number Generation} \label{Sec:connection}

\textchange{In this section, we shall point out a connection between the problem of fixed-length random number generation (FL-RNG),
and the variable-length random number generation (VL-RNG).\footnote{\textchange{Here, we fix the length of target random variables, and consider
RNG algorithms with fixed/variable length of coin random variables.}} As in the previous sections, let $\bm{X}=\{ X^m \}_{m=1}^\infty$ and $\bm{Y}= \{ Y^n \}_{n=1}^\infty$ be the coin and target processes.
An FL-RNG algorithm is described by a deterministic function $\psi: {\cal X}^m \to {\cal Y}^n$, and the approximation error is defined by
\begin{align*}
\Delta_m := \| P_{\tilde{Y}^n} - P_{Y^n} \|_1
\end{align*}
for $\tilde{Y}^n = \psi(X^m)$, where $\| P- Q \|_1 := \frac{1}{2} \sum_x |P(x)-Q(x)|$ is the variational distance
between two distributions $P$ and $Q$. }

\textchange{In the problem of source coding, it is recognized that
there is an intimate connection between the error probability of almost lossless 
fixed-length (FL) code and the overflow probability of the code length of variable-length (VL) code (eg.~see \cite{MerNeu:92, UchHan:01, kontoyiannis-verdu:14}).
More specifically, for a given VL code, we can construct a FL code such that the error probability is the same as
the overflow probability of the original VL code; and vice versa. 
In a similar vein, we can convert a given VL-RNG algorithm to an FL-RNG algorithm as follows.}
\textchange{
\begin{proposition} \label{proposition:connection-VL-FL}
For a given VL-RNG algorithm $\phi$ satisfying \eqref{eq:validity}, there exists an FL-RNG algorithm $\psi$ such that the approximation error satisfies 
\begin{align*}
\Delta_m \le \Pr\big( T > m \big),
\end{align*}
where $T$ is the stopping time of $\phi$.
\end{proposition} 
\begin{proof}
For the set ${\cal L}_\phi$ of all leaves, let ${\cal B}=\{ s \in {\cal L}_\phi : |s| \le m\}$.
Recall that, by our convention, we denote $\phi(x^m) = y^n$ if the algorithm terminates with output
$\phi(x^i)=y^n$ for some $i \le m$, and $\phi(x^m) = \bot$ otherwise. 
By using these notations, we set 
\begin{align*}
\psi(x^m) = \left\{
\begin{array}{ll}
\phi(x^m) & \mbox{if } \phi(x^m) \in {\cal Y}^n \\
b^n & \mbox{else}
\end{array}
\right.
\end{align*}
where $b^n \in {\cal Y}^n$ is an arbitrarily fixed sequence. 
Then, since 
\begin{align*}
P_{\tilde{Y}^n}(b^n) &= \sum_{s \in {\cal B} : \atop \phi(s)=b^n} P_{X^{|s|}}(s) + \sum_{s \in {\cal L}_\phi \backslash {\cal B}} P_{X^{|s|}}(s) \\
&\ge \sum_{s \in {\cal L}_\phi : \atop \phi(s)=b^n} P_{X^{|s|}}(s) \\
&= P_{Y^n}(b^n)
\end{align*}
and 
\begin{align*}
P_{\tilde{Y}^n}(y^n) &= \sum_{s \in {\cal B} : \atop \phi(s)=y^n} P_{X^{|s|}}(s) \\
&\le   \sum_{s \in {\cal L}_\phi : \atop \phi(s)=y^n} P_{X^{|s|}}(s) \\
&= P_{Y^n}(y^n)
\end{align*}
for every $y^n \neq b^n$, we have
\begin{align*}
\Delta_m &= P_{\tilde{Y}^n}(b^n) - P_{Y^n}(b^n) \\
&= \sum_{s \in {\cal L}_\phi \backslash {\cal B}: \atop \phi(s) \neq b^n} P_{X^{|s|}}(s) \\
&\le \sum_{s \in {\cal L}_\phi \backslash {\cal B}} P_{X^{|s|}}(s) \\
&= \Pr\big( T > m \big),
\end{align*}
where the first equality follows from an alternative expression of the variational distance (eg.~see \cite[Eq.~(11.137)]{cover}).
\end{proof}
}

\textchange{
As we can find from the proof of Proposition \ref{proposition:connection-VL-FL}, we can convert any VL-RNG algorithm to a FL-RNG algorithm 
just by stopping the VL-RNG algorithm after a prescribed number of coin tosses. In fact, the achievability bound for the FL-RNG \cite[Lemma 2.1.1]{han:book} 
can be also attained by the modified version of the interval algorithm via Proposition \ref{proposition:connection-VL-FL} and Theorem \ref{theorem:performance-interval} up to a negligible constant factor;
in the asymptotic regime, if we set $m = nR$ with $R > \overline{H}(\bm{Y})/\underline{H}(\bm{X})$, then Theorem \ref{theorem:stopping-interval-asymptotic} guarantees that
the approximation error $\Delta_m$ of the FL-RNG converges to $0$ (cf.~Theorem \cite[Theorem 2.1.1]{han:book}).
Conversely, even though we proved the converse bound for the VL-RNG (Theorem \ref{theorem:converse}) directly in Section \ref{sec:problem}, we can provide an alternative proof
by combining Proposition \ref{proposition:connection-VL-FL} and the converse bound for the FL-RNG in \cite[Lemma 2.1.2]{han:book};
in the asymptotic regime, the converse bound in Theorem \ref{theorem:stopping-interval-asymptotic}, i.e., $R \ge \max[\overline{H}(\bm{Y})/\overline{H}(\bm{X}), \underline{H}(\bm{Y})/\underline{H}(\bm{X})]$
for every achievable rate $R$, can be obtained from Proposition \ref{proposition:connection-VL-FL} and \cite[Theorem 2.1.2]{han:book}.
Unlike the source coding, 
the opposite claim, i.e., possibility of converting a FL-RNG to a VL-RNG,
is not clear in general. }
 
\section{Discussion} \label{sec:discussion}

In this paper, we revisited the problem of exactly generating a random process with another random process,
and proved the optimality of the interval algorithm when either the coin or the target process has 
one point spectrum. However, when both the coin and the target processes have spreading spectrum,
the achievability and the converse bounds derived in this paper do not coincide.
At least, there is room for improvement on the achievability bound;
for instance, when the coin process and the target process are identical and have spreading spectrum,
the interval algorithm apparently attains the unit rate but the upper bounds in Theorem \ref{theorem:stopping-interval-asymptotic} and Theorem \ref{theorem:asymptotic-average} are loose. 
In order to derive tighter bounds, instead of the upper and lower limits of the spectrums, we need to
analyze spreading spectrums more carefully. 
For the random number generation with approximation error, such a direction of
research was conducted in \cite{NagMiy:96, AltWag:12}.

\textchange{In a similar vein, the bounds derived in this paper may not be tight for finite block length regime in general.
When either the coin process or the target process is unbiased  and the other process is i.i.d.,
by an application of the central limit theorem to the bounds in Theorem \ref{theorem:converse} and Theorem \ref{theorem:performance-interval},
we can derive bounds that coincide up to the so-called second-order rate \cite{hayashi:09, polyanskiy:10}.
In other words, the interval algorithm is optimal up to the second-order rate in that case.
It is an important research direction to conduct the finite block length analysis of the interval algorithm
when both the coin and target processes are biased. 
It should be noted that, when the coin process is i.i.d., the average stopping time of the interval algorithm
is known to be tight up to ${\cal O}(1)$ term \cite{HanHoshi97}.}

Another important research direction is the interval algorithm with finite precision arithmetic.
In order to implement the interval algorithm, the updates of intervals must be conducted with
finite precision arithmetic in practice. Such a direction of research was conducted in \cite{UyeLi:03} for i.i.d. processes. 




\bibliographystyle{./IEEEtranS}
\bibliography{../../../09-04-17-bibtex/reference}

\end{document}